\documentclass[a4paper,10pt]{article}
 \usepackage[english]{babel}
  \usepackage{amsmath,amsfonts,amssymb,amsthm}

\usepackage{amsmath}
  \usepackage{paralist}
  \usepackage{graphics} 
  \usepackage{epsfig} 

 \usepackage[colorlinks=true]{hyperref}

\newtheorem{theorem}{Theorem}[section]
\newtheorem{corollary}{Corollary}

\newtheorem{lemma}[theorem]{Lemma}
\newtheorem{proposition}{Proposition}

\theoremstyle{definition}

\newtheorem{remark}{Remark}

\def\R{\mathbb{R}}

\let\e=\varepsilon
\let\.=\cdot
\def\dev{\operatorname{div}}
\def\id{i.~e.~}

\newcommand{\dk}{\nabla_{v_k}}
\newcommand{\dl}{\nabla_{v_l}}
\newcommand{\V}{\text{V}}

\newcommand{\fj}{f_j^N}
\def\cal{\mathcal}

\title{On the Kac model for the Landau equation}

\author{Evelyne Miot\footnote{miot@ann.jussieu.fr},  Mario Pulvirenti\footnote{pulvirenti@mat.uniroma1.it} and Chiara Saffirio\footnote{saffirio@mat.uniroma1.it}}

\begin{document}
\maketitle

\centerline{\scshape Evelyne Miot }
\medskip
{\footnotesize
 \centerline{Laboratoire de Math\'ematiques, Universit\'e Paris-Sud 11, b\^at. 425, }
   \centerline{91405 Orsay, France}
} 

\medskip

\centerline{\scshape Mario Pulvirenti }
\medskip
{\footnotesize
 \centerline{ Dipartimento di
Matematica Guido Castelnuovo, Universit\`a La Sapienza - Roma,}
   \centerline{ P.le A. Moro, 5 00185 Roma, Italy}
} %

\medskip

\centerline{\scshape Chiara Saffirio }
\medskip
{\footnotesize
 \centerline{ Dipartimento di
Matematica Guido Castelnuovo, Universit\`a La Sapienza - Roma,}
   \centerline{ P.le A. Moro, 5 00185 Roma, Italy}
} %

\bigskip


\begin{abstract}
We introduce a $N$-particle system which approaches, in the mean-field limit, the solutions
of the Landau equation with Coulomb singularity.
This model plays the same role as the Kac's model for the homogeneous Boltzmann equation.
We use compactness arguments following \cite {V1}.
\end{abstract}

\section{Introduction}
In 1954 M. Kac \cite{Kac}, in the attempt of clarifying some aspects of the transition
from particle systems to the Boltzmann equation,
introduced a toy model which has been successively investigated. See for instance \cite{pey} and references quoted therein.

Roughly speaking the Kac's model consists in a $N$-particle system.
The particles have no position but only velocities denoted by $V_N=(v_1,\ldots,v_N)\in \R^{3N}$.
The dynamics is the following stochastic process. At a random time, pick a pair of particles, say $i$ and $j$, and perform the transition
$$
v_i,v_j \to v_i', v_j'
$$
preserving total momentum and energy.

More precisely, if $W^N=W^N(V_N,t)$ is a symmetric probability distribution describing a statistical
state of the system, the time evolution is given by the following master equation
\begin{equation}
\label{Kac 1}
\begin{split}
\partial_t W^N ={\cal{L}}_N W^N
\end{split}
\end{equation}
where
\begin{equation}
\label{gen 1}
 \begin{split}
{\cal{L}}_N W^N=\frac{1}{2N}\sum_{i\neq j}&\int \,dv_i'  \,dv_j'\,   K(v_i,v_j| v_i',v_j')
\delta( v_i+v_j-v_i'-v_j') \delta( v_i^2+v_j^2-v_i'^2-v_j'^2 )\\
&\{W^N(v_1,\ldots, v_{i}',\ldots, v_{j}',\ldots, v_N)-W^N(v_1,\ldots, v_N)\},
\end{split}
\end{equation}
and $K$ is a suitable kernel.

Introducing the exchanged momentum $p=v_i'-v_i=v_j-v_j'$ in the collision process and assuming that
\begin{equation}
K(v_i,v_j| v_i',v_j')=w(p)
\end{equation}
for some smooth and radially symmetric $w$, we readily arrive to
\begin{equation}
\label{gen 2}
\begin{split}
{\cal{L}}_N W^N=\frac{1}{2N}\sum_{i\neq j}&\int dp\, w(p)
\delta\left( p^2-p\cdot (v_i-v_j)\right)\\
&\{W^N(v_1,\ldots, v_{i}+p,\ldots, v_{j}-p, \ldots, v_N)-W^N(v_1,\ldots, v_N)\}.
\end{split}
\end{equation}
In \cite{Kac} it was shown that  the first marginal of $W^N$
converges, in the limit $N \to \infty$, to
the solution to the (homogeneous) Boltzmann equation if the initial datum is chaotic, \id if $W^N(0)= f_0^{\otimes N}$
for some probability distribution $f_0$. Moreover,  the $j$-particle marginal converge to the $j$-fold product of such solution, i.e., propagation
of chaos holds (see \eqref{propa} below).

The main purpose of the present paper is to introduce an analogous model for
the Landau equation with Coulomb interaction.
A straightforward way to derive this model is to perform the so-called grazing collision limit
on eq.n \eqref{Kac 1} as we shall do in a moment.
In fact in 1936 Landau \cite{LL}, starting from the
Boltzmann
collision operator,  derived a new kinetic equation for the time evolution of a dense charged plasma,
exploiting the fact that,
in this physical context,
only the grazing collisions ($p\approx 0$) are relevant.
According to such a prescription, we introduce $\e>0$ a small parameter
and scale the kernel of ${\cal{L}}_N$ in eq.n  \eqref{gen 2}
as
$$
w(p) \to \frac 1{\e^3} w\left(\frac p {\e}\right)
$$
so that
\begin{equation}
\label{gen 3}
\begin{split}
{\cal{L}}_N^\e W^N=\frac{1}{2N\e^4 }\sum_{i\neq j}&\int \,dp\, w\left(\frac p {\e}\right)
\delta\left( p^2-p\cdot (v_i-v_j)\right)\\
&\{W^N(v_1,\ldots, v_{i}+p,\ldots, v_{j}-p, \ldots, v_N)-W^N(v_1,\ldots, v_N)\}.
\end{split}
\end{equation}
Note that we inserted another factor $1/\e$ in front of the collision operator, to take into account the large
density of the plasma.

Now, for fixed $N$, we perform the limit $\e \to 0$. By a straightforward formal computation (change of variables and Taylor expansion),
we readily detect the limiting generator which is the following diffusion operator:
\begin{equation}
\label{gen 4}
\begin{split}
\tilde{L}^N =\dev_{V_N}(B\cdot \nabla_{V_N}).
\end{split}
\end{equation}
Here
\begin{equation*}
 B : \R^{3N}\to \R^{3N\times3N}
\end{equation*}
is a matrix
 defined in the following way. For $V_N=(v_1,\ldots,v_N)\in \R^{3N}$,
\begin{equation*}
\begin{cases}
\displaystyle B_{i,j}(V_N)=-\frac{a(v_i-v_j)}{N} \ \ \ \text{if} \ \ i\neq j,\\
\displaystyle B_{i,i}(V_N)=\frac{1}{N}\sum_j a(v_i-v_j),
\end{cases}
\end{equation*}
where the $3\times3 $ matrix $a$ is given by
\begin{equation}\label{a}
a(w)=\frac{1}{|w|}(\mathbb{I}-\hat{w}\otimes\hat{w})=\frac{1}{|w|}P(w ),
\quad w\in \R^3,\quad \text{and } \hat{w}=\frac{w}{|w|},
\end{equation}
with $P(w)$ the orthogonal projection on the plane orthogonal to $w$.

Unfortunately the elliptic operator $\tilde{L}^N$ has two main disadvantages.
First it is not uniformly elliptic (see Lemma \ref{lemma:positive} below), second it is not smooth due to the
divergence for $|v_i-v_j|\approx 0$.

As a matter of fact, since we want a handier N-particle model to start with, we slightly modify $\tilde{L}^N$ to obtain
a smooth and non-degenerate operator. More precisely, we define
\begin{equation}\label{L^N}
L^N=\dev_{V_N}(B^N\nabla_{V_N})
\end{equation}
where $B^N$ is obtained by making the matrix $B$ smooth and bounded from below:
\begin{equation}\label{B_N}
\begin{cases}
\displaystyle B_{i,j} ^N (V_N)=-\frac{a^N(v_i-v_j)}{N} \ \ \ \text{if} \ \ i\neq j,\\
\displaystyle B_{i,i} ^N (V_N)=\frac{1}{N}\sum_j a^N(v_i-v_j)+\frac{1}{N}.
\end{cases}
\end{equation}
Here the $3\times3 $ matrix $a^N$ is obtained by replacing $\frac{1}{|w|}$ by
$\bar{\chi}_{\frac{1}{N}}(|w|)\frac{1}{|w|}$ in \eqref{a}, defining
\begin{equation}
\label{cut-off}
\chi_{\frac{1}{N}} \in C^\infty (\R^+),  \quad  \chi_\frac{1}{N}(r) =1 \quad \text {if}  \quad r<\frac{1}{N} ,
\quad \chi_{\frac{1}{N}} (r)=0 \quad \text {if}  \quad r> \frac{2}{N},
\end{equation}
and $\bar{\chi}_N=(1-\chi_N)$.
Now the evolution equation assumes the form
\begin{equation}\label{Kac 3}
\partial_t W^N=\dev_{V_N}(B^N\nabla_{V_N}W^N)
\end{equation}
and the well-known theory of linear parabolic equations assures the
existence of a unique classical solution for $L^1$ initial data.

To simplify the notations we define
\begin{equation*}
\frac{1}{|w|_N} := \bar{\chi}_{\frac{1}{N}}(|w|)\frac{1}{|w|}
\end{equation*}
so that $a^N(w)=\frac{1}{|w|_N}P(w)$.

\medskip

In the limit $N\to \infty$, the number of variables in the definition of $W^N$ diverges, hence we will actually prefer to
look at the asymptotic behavior of the \emph{marginal distributions}
\begin{equation*}
 f_j^N(v_1,\ldots,v_j,t)=\int\,dv_{j+1}\ldots dv_N\, W^N(v_1,\ldots,v_N,t),\quad j=1,\ldots,N.
\end{equation*}
Note that $f_N^N=W^N$ and the $j$-th marginal distribution is a function of $j$ variables.
Moreover, using  \eqref{Kac 3} we can express
the evolution of each $\fj$ in terms of $f_{j+1}^N$. Straightforward computations
lead to the following system of equations, called the $N$-particle \emph{hierarchy}
\begin{equation}
\label{eq:hierarchy-N}
\partial_t \fj=L_j^N\fj+\frac {N-j} N C_{j+1} ^N f_{j+1}^N,\quad j=1,\ldots,N-1
\end{equation}
where $L_j^N$ and $C_{j+1} ^N$ are operators defined  by:
\begin{equation}
\label{coll}
\begin{split}
 L_j ^Nf_j ^N&
=\frac{1}{N}\sum_{\substack{k\neq l\\k,l=1}}^j\dk\cdot\left[a_{k,l} ^N\cdot (\dk \fj-\dl\fj)\right]+\frac{1}{N}\sum_{\substack k=1} ^j \Delta_{v_k}f_j ^N,\\
C_{j+1}^N f_{j+1} ^N&
=\sum_{k=1}^{j}\dk\cdot\int dv_{j+1}\,a_{k,{j+1}} ^N\cdot \left(\dk f_{j+1} ^N-\nabla_{v_{j+1}}f^{N} _{j+1}\right).
\end{split}
\end{equation}
In particular we have $L_N^N=L^N$.

Since $C_j=O(1)$, while $ L_j ^N f_j ^N=O(\frac j N)$, the formal limit of \eqref{eq:hierarchy-N}
as $N\to \infty$ yields an infinite system of equations called \emph{Landau hierarchy}
\begin{equation}
\label{eq:hierarchy}
 \partial_t f_j=C_{j+1}f_{j+1},\quad j=1,\ldots,+\infty,
\end{equation}
where the operators $C_{j+1}$ write
\begin{equation*}
 C_{j+1} g
=\sum_{k=1}^{j}\dk\cdot\int dv_{j+1}\,a_{k,{j+1}} \cdot \left(\dk g-\nabla_{v_{j+1}}g \right).
\end{equation*}

Due to the structure of the collision operator $C_{j+1}$, we realize that special solutions to eq.n \eqref{eq:hierarchy}
are given by factorized states
\begin{equation}\label{propa}
f_j (v_1 \ldots v_j,t)=\prod _{i=1}^j f (v_i,t)=f(t)^{\otimes j}
\end{equation}
where the one particle distribution $f(t)$ solves the Landau equation
\begin{equation}
\label{eq:landau}
\begin{split}
\partial_t f=Q(f,f),
\end{split}
\end{equation}
with
\begin{equation}\label{gen4}
Q(f,f)(v)=\int_{\R^3} dw \,a(v-w)\cdot (f(w)\nabla f(v)-f(v)\nabla f(w)).
\end{equation}
It should be mentioned that, conversely, if $f$ is a solution to eq.n \eqref{eq:landau}, then the products $f_j=f^{\otimes j}$ solve the hierarchy
\eqref{eq:hierarchy}.

\medskip

Following the general paradigm of the kinetic theory, we expect that propagation of chaos holds,
namely that \eqref{propa} holds for all time provided that the initial state is chaotic, i.e. \eqref{propa} is initially verified.

Actually, we are not able to show propagation of chaos. We will be able to prove only the (weak) convergence
$f_j ^N(t)  \to f_j (t)$ (for suitable subsequences), being $f_j(t)$ a weak solution of the
Landau hierarchy \eqref{eq:hierarchy}, without knowing whether $f_j (t)$ factorizes even though it does at time zero.
The reason is that we have a poor control on the limiting hierarchy as well as on
the Landau equation \eqref{eq:landau}. In fact, we will obtain a solution to eq.n \eqref{eq:hierarchy} by adapting, to the present
$N$-particle context, a strategy, based on compactness arguments, introduced by C. Villani \cite{V1}
for the Landau equation. As a matter of fact we do not have uniqueness, which is a necessary condition to get propagation of
chaos. Indeed, assume that $f(t)$ and $g(t)$ are two weak solutions to eq.n \eqref{eq:landau}, with the same initial datum $f_0$.
It follows that
$$
f_j(t)=\lambda f(t)^{\otimes j}+(1-\lambda) g(t)^{\otimes j},\quad \lambda \in( 0,1)
$$
solves the Landau hierarchy with the chaotic initial datum $f_0^{\otimes j}$, but does not factorize.

\medskip

Before stating our main result, we make some assumptions on the initial a.c. measures $W^N(0)$:

\begin{enumerate}
\item $W^N(0)\geq 0$;
\item $W^N(0)$ is symmetric in the variables $v_1,\ldots,v_N$;
\item The following uniform bounds hold
\begin{equation*}
\begin{split}
\hspace*{-2em}&\int\,dV_N\, W^N(0)=1,\\
\hspace*{-2em}&\frac{1}{N}\int \,dV_N\, W^N(0)\log(W^N(0))\leq C,\\
\hspace*{-2em}&\frac{1}{N}\int \,dV_N\,W^N(0)|V_N|^2\leq C.
\end{split}
\end{equation*}
\end{enumerate}
These properties still hold true at positive times. Actually
$$
\int\,dV_N\, W^N(t)|V_N|^2=\int \,dV_N\,W^N(0)|V_N|^2+\frac{C}{N}t
$$
expresses the energy dissipation and follows easily by an integration by parts in eq.n \eqref{Kac 3}.
Moreover
$$
\int \,dV_N\,W^N(t)\log(W^N(t))\leq \int\,dV_N\, W^N(0)\log(W^N(0))
$$
expresses the entropy dissipation and will be discussed in the next section.

We now explain  what is the sense  we give to eq.n \eqref {eq:hierarchy}.
The main difficulty related to the Landau equation is due to the divergence of
the matrix $a(w)$ when $|w|$ is small. Indeed if $f_{j+1}$  (some weak limit of
$f^N_{j+1}$)  is only  in $L^1(\R^{3(j+1)})$,
the integral
$$
\int f_{j+1}(v_1,\ldots,v_{j+1}) \frac 1 {|v_i-v_{j+1}|}
$$
makes no sense; therefore $C_{j+1} f_{j+1} $ is not defined in general.
Thus, as we did before in \eqref{cut-off} to regularize the operator $\tilde{L}^N$,
we introduce a small parameter $\delta>0$ and the cut-off function $\chi_\delta \geq 0$, not increasing and  such that
\begin{equation}
\label{cutoff}
\chi_\delta \in C^\infty (\R^+),  \quad  \chi_\delta(r) =1 \quad \text {if}  \quad r<\delta ,
\quad \chi_\delta (r)=0 \quad \text {if}  \quad r> 2\delta.
\end{equation}

Then we define $C^\delta _{j+1}$ replacing $a(w)$ in definition \eqref{coll}  by
$a (w)(1-\chi_\delta (|w|)$, thus removing the singularity. Clearly, if $\varphi\in C_c^2$ then
$\int \varphi C_{j+1}^\delta f_{j+1} $ makes sense for any $f_{j+1}\in \mathcal{M}(j+1)$, where
$\mathcal{M}(k),k\geq 0,$ denotes the space of probability measures on $\R^{3k}$ equipped with the topology given by the weak convergence of probability measures.

Our result can be stated as follows
\begin{theorem}
\label{thm:main}
There exists a subsequence $N_k\to \infty$ such that, for all $j$,
there exists $f_j\in L^\infty([0,T]; L^1)\cap C^0([0,T] ; \mathcal{M}(j))$, with
finite mass, energy and entropy, such that
\begin{equation*}
 f_j^{N_k}\to f_j\quad \text{when}\quad k\to  \infty,
\end{equation*}
where the convergence holds in the sense of weak convergence of probability measures.
For any $t>0$ and for any test function $\varphi \in C^2_c (\R^{3j})$, the limit
\begin{equation*}
\lim_{\delta \to 0}
 \int _0^t\,ds \int dv_1\dots\,dv_j \,\varphi(v_1,\ldots,v_j)  C^\delta_{j+1} f_{j+1}(v_1,\ldots,v_j,s),\quad j=1,\ldots,+\infty
\end{equation*}
exists, and we have
\begin{equation*}
\int  \varphi f_j(t)-\int \varphi f_j(0)= \int _0^t ds\int \varphi C_{j+1} f_{j+1}(s),\quad j=1,\ldots,\infty.
\end{equation*}

\end{theorem}

\begin{remark} Following \cite{V1}, as we shall see in the proof of Theorem \ref{thm:main} we  have more regularity on
 $f_j$ (see \eqref{ineq:p}). This allows us to give a direct sense to $C_{j+1}$ without using a cut-off function.
\end{remark}

 We conclude this section with some additional remarks.

Another kind of  Landau equations can also be considered replacing the matrix $a$
by
$$
a_\alpha (w)=\frac{1}{|w|^\alpha}(\mathbb{I}-\hat{w}\otimes\hat{w}),
$$
with $\alpha<1$.
In case of $\alpha<0$ a unique smooth solution can be constructed (see \cite {DV}, \cite{DV1}). It would be interesting to consider
a $N$-particle diffusion process with generator given by
\eqref{gen4}, in which $a$ is replaced by $a_\alpha$. Of course now we expect a much better control
on the limit $N\to \infty$ and, in particular,  the propagation of chaos.

The Landau equation can also be obtained as a grazing collision limit  from the
homogeneous Boltzmann equation, for a sufficiently small $\alpha$ (see \cite{AB},  \cite {G},  \cite {DV} and  \cite {DV1}).
The case $\alpha=1$ has been considered in \cite {V1}.

In this paper we focus our attention on the Coulomb divergence $\alpha=1$,  which we think is
the most physically relevant case.
Indeed the Landau equation for $\alpha=1$ is believed to hold in the so called weak-coupling limit, for
Hamiltonian particle systems interacting by means of a smooth, short-range
potential. See \cite{Bal} and \cite {MP}  for a formal derivation.
Unfortunately up to now no rigorous result is known, even for short times.

\section{Proof of Theorem \ref{thm:main}}

\subsection{Preliminaries}

In this section, we collect some preliminary properties satisfied by the $N$ marginal distributions $\fj,j=1,\ldots,N$.
In all this section $N$ is fixed. We start by introducing some

\noindent \textbf{Notations}.
In the following, we will write
\begin{equation*}
 \V_j=(v_1,\ldots,v_j),\quad \V_j ^N=(v_{j+1},\ldots,v_N),\quad j=1,\ldots,N,
\end{equation*}
so that
\begin{equation*}
\fj=\fj(\V_j,t)=\int dV_j ^N \, W^N(\V_j,V_j ^N,t).
\end{equation*}
Moreover,
\begin{equation*}
 a_{i,j}=a_{i,j}(V_N)=a(v_i-v_j),\quad P_{i,j}=P(v_i-v_j),\quad i,j=1,\ldots,N.
\end{equation*}
$'' \cdot ''$ will denote the usual scalar product on $\R^3$, $\R^{3j}$ or $\R^{3N}$.
For $V_N,\xi \in\R^{3N}$,
\begin{equation*}
B(V_N)\cdot \xi=\begin{pmatrix} B_1(V_N)\cdot \xi \\ \cdot\\ \cdot \\ \cdot \\ B_N(V_N)\cdot \xi \end{pmatrix}
\end{equation*}
where $B_k(V_N)\in \R^{3N}$ is the $k$-th line of $B(V_N)$.

On the other hand, for $1\leq k\leq N$ we will denote by
\begin{equation*}
 \dk\cdot \xi=\sum_{i=1}^3 \partial_{v_{k}^i} \xi_i,
\end{equation*}
where $\xi=(\xi_1,\xi_2,\xi_3)$ and $v_k=(v_{k}^1,v_{k}^2,v_{k}^3)$.

Finally, for every fixed $j$ such that $1\leq j\leq N$, for $1\leq k,m\leq j$ we denote by
\begin{equation*}
V_j^{k,m}=(v_1,\ldots,v_{k-1},v_{m},v_{k+1},\ldots,v_{m-1},v_{k},v_{m+1},\ldots,v_j)
\end{equation*}
the vector obtained by exchanging the components $v_k$ and $v_m$.

We start with an elementary property on the matrix $B$.
\begin{lemma}
\label{lemma:positive}
$B$ is positive semi-definite, \id for all $\xi$
$$(B\cdot\xi)\cdot\xi\geq0.$$
More precisely, we have
\begin{equation*}
 (B\cdot\xi)\cdot\xi=\frac{1}{N}\sum_{i,j=1}^{N} \frac{|P_{i,j}\cdot (\xi_i-\xi_j)|^2}{|v_i-v_j|},\quad
\text{where} \quad \xi=(\xi_i)_{1\leq i\leq N}.
\end{equation*}
\end{lemma}
\begin{proof}
Fix $\xi\in\R^{3N}$, setting conventionally  $a_{i,i}=0$ for all $i$ we get
\begin{equation*}
 \begin{split}
(B\cdot \xi)\cdot\xi
&=\sum_{i=1} ^{N}\left(-\frac{1}{N}\sum_{j\neq i} a_{i,j}\cdot \xi_j+\frac{1}{N}\sum_{j} a_{i,j}\cdot \xi_i\right)\cdot\xi_i\\
&=\frac{1}{N}\sum_{i,j=1}^{N} \frac{P_{i,j}\cdot (\xi_i-\xi_j)}{|v_i-v_j|}\cdot \xi_i.
\end{split}
\end{equation*}
Exchanging $i$ and $j$ in the sum we get, using that $P_{i,j}$ is a projector :
\begin{equation*}
 \begin{split}
  (B\cdot \xi)\cdot\xi
&=\frac{1}{N}\sum_{i,j=1}^{N} \frac{P_{i,j}\cdot (\xi_i-\xi_j)}{|v_i-v_j|}\cdot (\xi_i-\xi_j)\\
&=\frac{1}{N}\sum_{i,j=1}^{N} \frac{|P_{i,j}\cdot (\xi_i-\xi_j)|^2}{|v_i-v_j|}\geq 0.
 \end{split}
\end{equation*}
\end{proof}

\begin{lemma}
 \label{lemma:H-theorem}
 Let $W^N(t)$ be the solution to eq.n \eqref{Kac 3}. Then
for any convex function $\Phi\in\mathcal{C}^2(\R^+;\R)$,  $\int dV_N \, \Phi(W^N)$ is decreasing in time; more precisely, we have
\begin{equation}
\frac{d}{dt}\int dV_N\,\Phi(W^N(t))=-\int dV_N\,\Phi''(W^N(t))\nabla_{V_N}W^N\.(B^N\cdot \nabla_{V^N}W^N)\leq0.
\end{equation}
\end{lemma}

\begin{proof}Look at
$$\partial_t W^N=L ^N W^N.$$
Let us consider a convex function $\Phi$, then
\begin{equation}\label{eq:H-theorem}
\begin{split}
\frac{d}{dt}\int \Phi(W^N)&=\int dV_N\,\Phi'(W^N)\dev_{V_N}(B^N\cdot \nabla_{V_N}W^N)\\
    &=-\int dV_N\,\Phi''(W^N)\nabla_{V_N}W^N\cdot  (B^N\cdot \nabla_{V_N}W^N).
\end{split}
\end{equation}
Taking into account the convexity of $\Phi$ and using Lemma \ref{lemma:positive} the r.h.s. of \eqref{eq:H-theorem} is non positive and the statement of the Lemma holds.
\end{proof}

In particular, we will use Lemma \ref{lemma:H-theorem} with $\Phi(x)=x\log(x)$. We denote by
\begin{equation}
S(W^N(t))=\frac{1}{N}\int dV_N \,W^N(t)\log(W^N(t))
\end{equation}
the entropy per particle. In view of Lemma \ref{lemma:H-theorem}, $S(W^N(t))$ is decreasing in time
\begin{equation}
\frac{d}{dt} S(W^N(t))=-\frac{1}{N}\int dV_N\,\frac{1}{W^N}\nabla_{V_N}W^N\cdot (B^N\cdot \nabla_{V_N}W^N)\leq 0
\end{equation}
since $\Phi''(x)=1/x\geq 0 $. In what follows we will use the explicit formula for the entropy production:
\begin{equation}
\label{diss-entropy}
\begin{split}
-\frac{d}{dt}S(W^N(t))&=\frac{1}{N^2}\sum_{k,l=1}^N\int dV_N\, \frac{|P_{k,l}
\cdot \left[\dk W^N-\dl W^N\right]|^2}{W^N|v_k-v_l|_N}\\
&+\frac 1{N^2} \int dV_N \, \frac 1 {W^N} | \nabla_{V_N} W^N|^2.
\end{split}
\end{equation}

\begin{remark}
Although the entropy $S(W^N(t))$ decreases,
$$ S(\fj(t))\equiv \frac{1}{j}\int f_j(t)\log(f_j(t))$$
is not  decreasing in general. However by subadditivity of the entropy we know (see e.g. \cite{Khin}) that
\begin{equation}
S(f_j^N(t))\leq S(W^N(t))
\end{equation}
so that
\begin{equation}
S(f_j^N(t))\leq C
\end{equation}
since we have $S(W^N(0))\leq C$.
\end{remark}

\begin{remark}
In case of factorization, \id $f_j=f^{\otimes j}$, we have the equality
\begin{equation}
S(f_j)=S(f).
\end{equation}
\end{remark}

\bigskip

Eq.n \eqref{diss-entropy} provides a useful estimate given by the following
\begin{corollary}
 \label{cor:gradient-N}
Let $0\leq s_1\leq s_2$. Then
\begin{equation*}
\sum_{k,l=1}^N  \int_{s_1}^{s_2} ds\int dV_N\,\frac{|P_{k,l}\cdot \left[\dk W^N-\dl W^N\right]|^2}{W^N|v_k-v_l|_N} \leq CN^2.
\end{equation*}
\end{corollary}

\begin{remark}
Due to the symmetry of $W^N$, all terms of the above sum  are equal  and hence each term is bounded uniformly in $N$, namely for all $1\leq k,l\leq N$
\begin{equation}\label{bound-entropy}
\int ds\int dV_N\, \frac{|P_{k,l}
\cdot \left[\dk W^N-\dl W^N\right]|^2}{W^N|v_k-v_l|_N}\leq C.
\end{equation}
\end{remark}

\subsection{Basic estimates}

\begin{proposition}
\label{prop:L}
 Let $1\leq j\leq N-1$ and $\varphi\in C_c^2(\R^{3j},\R)$ be a test function. Let $0\leq s_1\leq s_2$. Then
\begin{equation*}
\int_{s_1}^{s_2} ds\,\left|\int dV_j\,L_j^N \,\fj(V_j)
\varphi(V_j)\right|\leq \frac{C(\varphi)j^2}{N}|s_1-s_2|^{1/2}
\end{equation*}
and
\begin{equation*}
\int_{s_1}^{s_2}ds\,\left| \int dV_j\,C_{j+1}^N\,f_{j+1}^N(V_j)
\varphi(V_j)\right|\leq C(\varphi) j |s_1-s_2|^{1/2},
\end{equation*}
where $C(\varphi)$ depends only on $\varphi$ and on the initial data, but not on $N$.
\end{proposition}

\begin{proof}
 We begin by estimating $C_{j+1}^N$. Recall \eqref{coll}.
 By integrating by parts, we have
\begin{equation*}
 \begin{split}
\int dV_j\, &C_{j+1}^N\,f_{j+1} ^N (V_j) \varphi(V_j)\\
&=-\sum_{\substack{k=1}}^j\int dV_j\,dV_j ^N \,a^N(v_k-v_{j+1})\cdot
(\dk W^N-\nabla_{v_{j+1}} W^N)(V_j,\V_j ^N)\cdot \dk \varphi(V_j)\\
&=\frac{1}{2}\sum_{\substack{k=1}}^j
\int  dV_N \, a^N(v_k-v_{j+1})\cdot (\dk W^N-\nabla_{v_{j+1}} W^N)(V_N)\cdot
\\&(\dk \varphi(V_j)-\nabla_{v_{k}} \varphi(V_j ^{k,j+1})),
 \end{split}
\end{equation*}
where
\begin{equation*}
V_j ^{k,j+i}=(v_1,\ldots,v_{k-1},v_{j+1},v_{k+1},\ldots,v_j)
\end{equation*}
and we exchanged variables $v_k$ and $v_{j+1}$ in the second line and used the symmetry of $W^N$.

Therefore
\begin{equation*}
 \begin{split}
 \int_{s_1}^{s_2}ds\,\Big|& \int dV_j\,C_{j+1}^N\,f_{j+1}^N(V_j) \varphi(V_j)\Big|\\
 &\leq\frac{1}{2}\int_{s_1}^{s_2}ds\,
\sum_{\substack{k=1}}^j \int dV_N\,\frac{\sqrt{W^N}}{\sqrt{W^N}}
\frac{|\dk \varphi(V_j)-\nabla_{v_{k}} \varphi(V_j ^{k,j+1})|}{\sqrt{|v_k-v_{j+1}|_N}}\\
 &\frac{|P_{k, j+1}\cdot(\dk W^N-\nabla_{v_{j+1}} W^N)(V_N)|}{\sqrt{|v_k-v_{j+1}|_N}};
 \end{split}
\end{equation*}
(by using the Cauchy-Schwarz inequality)
\begin{equation*}
\begin{split}
&\leq\frac{1}{2}\sum_{\substack{k=1}}^j\left(
\int_{s_1} ^{s_2}ds\int dV_N\,W^N(V_N)
\frac{|\dk\varphi(V_j)-\nabla_{v_{k}}\varphi(V_j ^{k,j+1})|^2}{|v_k-v_{j+1}|_N}\right)^{1/2}\cdot\\
&\left(\int_{s_1} ^{s_2}ds\,\int dV_N\,
\frac{|P_{k,j+1}(\dk W^N(V_N)-\nabla_{v_{j+1}} W^N(V_N))|^2}{W^N(V_N)|v_k-v_{j+1}|_N}\right)^{1/2}.
\end{split}
\end{equation*}
By virtue of mean-value Theorem applied to $\nabla_{v_k}\varphi$ and \eqref{bound-entropy} we get the bound on $C_{j+1}^N$:
\begin{equation}
\begin{split}
\int_{s_1} ^{s_2}ds\int C_{j+1}^N f_{j+1}^N\varphi(V_j)\,dV_j\leq j \, C(\varphi) |s_1-s_2|^{1/2}.
\end{split}
\end{equation}

By performing exactly the same computations we are led to
\begin{equation*}
 \begin{split}
\int_{s_1}^{s_2}ds\,\Big|& \int dV_j\,L_j^N\fj(V_j) \varphi(V_j)\Big|\\
&\leq \frac{C(\varphi)}{N}
\sum_{\substack{k\neq l\\k,l=1}}^j \left(\int_{s_1}^{s_2} ds
\int dV_N \frac{|P_{k,l}\cdot \left[\dk W^N(V_N)-\dl W^N(V_N)\right]|^2}
{W^N|v_k-v_l|_N}\right)^{1/2}\\
&\leq \frac{C(\varphi)j^2}{N}|s_1-s_2|^{1/2}.
\end{split}
\end{equation*}
The proof is now complete.
\end{proof}

\subsection{Convergence}

In this subsection, we establish the weak
compactness for the $\fj$ by making use of the uniform estimates established in the previous
subsection.
\begin{proposition}
\label{prop:compactness}
Let $\fj$ satisfy the hierarchy \eqref{eq:hierarchy-N}. There exists a
subsequence $N_k\to +\infty$ such that for any fixed $j$, there exists $f_j=f_j(V_j,t)\in C([0,T];  \mathcal{M}(j))$, with finite
energy and entropy, such that
 $f_j^{N_k}$ converges to $f_j$ weakly in the sense of measures, locally uniformly in time.
\end{proposition}

\begin{proof}

We fix $j$. For $\varphi\in C_c(\R^{3j})$, we set
\begin{equation*}
t\mapsto  g_\varphi^N(t)=\int\, dV_j\, f_j^N(V_j,t)\varphi(V_j).
\end{equation*}
We obtain a uniformly bounded sequence of
functions on $\R_+$. Moreover, when $\varphi\in C_c^2(\R^{3j})$,
by virtue of the proof of Proposition \ref{prop:L} the function $g_{\varphi}^N$ is uniformly equicontinuous.
Hence, by Ascoli's theorem and density of $ C_c^2(\R^{3j})$ in $C_c(\R^{3j})$,
there exists a subsequence $N_k$ such that for all $\varphi\in C_c(\R^{3j})$, $g_\varphi^{N_k}$ converges locally uniformly in time to some function $g_{\varphi}(t)$.  Now, for each fixed $t$, the map
\begin{equation*}
\varphi\mapsto g_\varphi(t)
\end{equation*}
is a positive linear form on $ C_c(\R^{3j})$. Thus the Riesz representation theorem
ensures the existence of a measure $df_j(t)$ such that
$g_\varphi(t)=\int \varphi df_j(t)$.
On the other hand,  $(f_j^N)(t)$ has
uniformly bounded entropy and energy; therefore it is weakly relatively compact in $L^1$.
This shows that in fact $df_j(t)=f_j(t)\,dV_j$ is an absolutely continuous probability measure and has finite entropy and energy. This concludes the proof of the proposition.

\end{proof}

\subsection{End of the proof}

We are now in position to complete the proof of Theorem \ref{thm:main}.
We fix $j\geq 0.$ For any $g\in C_c^2(\R^{3(j+1)})$ we set
\begin{equation*}
\begin{split}
C_{j+1}^\delta  g (V_j)&=\sum_{k=1}^j \dk\cdot \int  [(1-\chi_\delta)a](v_k-v_{j+1})\cdot
(\dk g-\nabla_{v_{j+1}}g)(V_j,v_{j+1})\,dv_{j+1},\\
 \overline{C}_{j+1}^\delta g (V_j)&=\sum_{k=1}^j \dk\cdot \int  [\chi_\delta a](v_k-v_{j+1})\cdot
(\dk g-\nabla_{v_{j+1}}g)(V_j,v_{j+1})\,dv_{j+1},
\end{split}
\end{equation*}
so that
\begin{equation}\label{C-delta}
 C_{j+1}(g)=C_{j+1}^\delta(g)+\overline{C}_{j+1}^\delta(g).
\end{equation}
The analogous decomposition holds for $C_{j+1} ^N$:
\begin{equation*}
C_{j+1} ^N=C_{j+1}^{N,\delta}+\bar{C}_{j+1}^{N,\delta}
\end{equation*}
where $a^N$ replaces $a$ in \eqref{C-delta}.
Note that $C_{j+1}^{N,\delta}=C_{j+1}^\delta$ whenever $N$ is sufficiently large.

We will show that for all $t\geq 0$ and for all test function $\varphi$ in $C_c^2$ we have
\begin{equation} \label{conv:smooth}
    \begin{split}
\int_0^t ds \int& \,dV_j\,C_{j+1}^{N,\delta}f_{j+1} ^N\varphi \\
&=\int_0^t ds \int \,dV_j\,C_{j+1}^\delta f_{j+1}^N\varphi \longrightarrow
 \int_0^t ds \int \,dV_j\,C_{j+1}^\delta f_{j+1}\varphi
    \end{split}
\end{equation}
when $N\rightarrow\infty$ and
\begin{equation}\label{ineq:reminder}
 \sup_{N\geq j}\left |\int_0^t ds \int \,dV_j\, \overline{C}_{j+1}^{N,\delta}f_{j+1}^N\varphi \right|
\leq C(\varphi)\delta^{1/2}.  \end{equation}

\medskip

First, \eqref{conv:smooth} follows by the convergence established in Proposition \ref{prop:compactness} and by two
integrations by parts.

\medskip

As regards \eqref{ineq:reminder}, we need a symmetrized form as in the proof of Proposition \ref{prop:L}. Mimicking the
computations of Proposition \ref{prop:L} we find
\begin{equation*}
 \begin{split}
&\left|\int_0^t ds \int dV_j\,\overline{C}_{j+1}^{N,\delta}f_{j+1}^N\varphi \right|
=\left|\int_0^t ds \int dV_j\,\overline{C}_{j+1}^\delta f_{j+1}^N\varphi \right|\\
&\leq C \sum_{k=1}^j\left(\int_0^t ds \int dV_N \,\frac{|P_{k,j+1}\cdot
(\dk W^N-\nabla_{v_{j+1}} W^N)|^2}{W^N|v_k-v_{j+1}|_N}\right)^{1/2}\\
&\quad \left(\int_0^t ds \int dV_N\,\chi_\delta^2(|v_k-v_{j+1}|)
\frac{|\dk \varphi(V_j)-\dk \varphi(V_j^{k,j+1})|^2}{|v_k-v_{j+1}|}W^N\right)^{1/2}.
 \end{split}
\end{equation*}
Applying once more inequality \eqref{bound-entropy}, the first term in the right-hand side is bounded.
 Next, we observe that in view of the support properties of $\chi_\delta$, the mean-value theorem
yields
\begin{equation*}
\begin{split}
\chi_\delta^2(|v_k-v_{j+1}|)|\dk \varphi(V_j)-\dk \varphi(V_j^{k,j+1})|^2\leq C\delta |v_k-v_{j+1}|.
\end{split}
\end{equation*}

Finally we obtain
\begin{equation*}
 \left|\int_0^t ds \int dV_j\, \overline{C}_{j+1}^\delta f_{j+1}^N\varphi \right|\leq C\delta^{1/2},
\end{equation*}
and \eqref{ineq:reminder} follows. Hence the proof of Theorem \ref{thm:main} is complete.

 \medskip

\medskip

We conclude this section with
some comments concerning additional regularity for the marginal $f_j^N$. In fact,
the control on the production of the total entropy (see Corollary \ref{cor:gradient-N})
yields a uniform control on the gradients of $f_j^N$. More precisely, we have for all $1\leq k,l\leq j$
\begin{equation}\label{ineq:profecj}
 \int ds \int dV_j\, \frac{|P_{k,l}
\cdot \left(\dk f_j^N-\dl f_j^N\right)|^2}{f_j^N|v_k-v_l|_N} \leq C.
\end{equation}
Indeed, we have
\begin{equation*}
\begin{split}
 \int ds &\int dV_j \, \frac{|P_{k,l}
\cdot \left(\dk f_j^N-\dl f_j^N\right)|^2}{f_j^N|v_k-v_l|_N}
\\&=\int ds \int dV_j\,  \frac{1}{f_j^N|v_k-v_l|_N}\left|\int P_{k,l}
\cdot \left(\dk W^N-\dl W^N\right)\,dV^N_j \right|^2\\
&=\int ds \int dV_j\, \frac{f_j^N}{|v_k-v_l|_N}\left|\int P_{k,l}
\cdot \left(\dk W^N-\dl W^N\right)\frac{1}{W^N}\,\frac{W^N}{f_j^N}\,dV^N_j \right|^2\\
&\leq \int ds \int dV_N\, \frac{f_j^N}{|v_k-v_l|_N}\int |P_{k,l}
\cdot \left(\dk W^N-\dl W^N\right)|^2\frac{1}{(W^N)^2}\,\frac{W^N}{f_j^N}\\
&=\int ds \int dV_N\, \frac{|P_{k,l}
\cdot \left(\dk W^N-\dl W^N\right)|^2}{W^N|v_k-v_l|_N},
\end{split}
\end{equation*}
where we have applied Jensen's inequality in the last inequality. The conclusion follows from \eqref{bound-entropy}.

\medskip

In particular, \eqref{ineq:profecj} implies that
\begin{equation*}
\frac{ P_{k,l}}{|v_k-v_l|_N}\cdot(\nabla_{v_k}\sqrt{f_j^N}-\nabla_{v_l}\sqrt{f_{j}^N})
\end{equation*}
is bounded in $L^2(\R_+\times\R^{3j})$; hence, following the same arguments as in \cite{V1} we can conclude that
\begin{equation}\label{ineq:p}
\frac{ P_{k,l}}{|v_k-v_l|}\cdot(\nabla_{v_k}\sqrt{f_j}-\nabla_{v_l}\sqrt{f_{j}})\in L^2(\R_+\times \R^{3j}),
\end{equation}
so that one can use the symmetrized form already used in the proof of Proposition \ref{prop:L} to define
$C_{j+1}f_{j+1}$ as in \cite{V1}:
\begin{equation*}
\begin{split}
\int& ds  \int \,dV_j\, C_{j+1}f_{j+1}\varphi
\\=-&\frac{1}{2}\sum_{k=1}^j\int ds \int \,dV_j\,a_{k,j+1}\cdot
\left(\dk f_{j+1}-\nabla_{v_{j+1}} f_{j+1}\right)\cdot \left(\dk \varphi(V_j)-\dk \varphi(V_j^{k,j+1})\right).
\end{split}
\end{equation*}

\section*{Acknowledgements} One of the authors (MP) thanks L. Desvilletes for useful and illuminating
discussions.

\medskip

\medskip

\end{document}